\documentclass[11pt]{article}

\usepackage{textcomp}
\usepackage{stmaryrd}

\usepackage{hyperref}
\usepackage{fullpage}
\usepackage{amsfonts}
\usepackage{amssymb}
\usepackage{amsmath}
\usepackage{latexsym}
\usepackage{amsthm}
\usepackage{eepic}
\usepackage{sectsty}

\newcommand {\Rahul} {Rahul Jain}
\newcommand {\Attila} {Attila Pereszl\'{e}nyi}
\newcommand {\Penghui} {Penghui Yao}
\newcommand {\CQTCS} {
    Centre for Quantum Technologies and Department of Computer Science,
    National University of Singapore.
}
\newcommand {\CQT} {
    Centre for Quantum Technologies,
    National University of Singapore.
}

\newcommand {\PdfTitle} {A direct product theorem for  bounded-round
    public-coin randomized communication complexity}
\hypersetup{
    pdftitle={\PdfTitle},
    pdfauthor={\Rahul, \Attila, \Penghui}
}

\theoremstyle{plain}
\newtheorem{theorem}{Theorem}[section]
\newtheorem{lemma}[theorem]{Lemma}

\newtheorem{cor}[theorem]{Corollary}
\newtheorem*{maintheorem}{Theorem~\ref{thm:direct product}}
\theoremstyle{definition}
\newtheorem{definition}[theorem]{Definition}
\newtheorem{claim}[theorem]{Claim}

\newtheorem{fact}[theorem]{Fact}

\newenvironment{proofof}[1]{\noindent{\bf Proof of #1:}}{\qed\\}

\newcommand{\defeq}{\stackrel{\mathrm{def}}{=}}

\def\({\left(}
\def\){\right)}

\def\X{\mathcal{X}}
\def\Y{\mathcal{Y}}
\def\Z{\mathcal{Z}}

\def\R{\mathcal{R}}
\def\P{\mathcal{P}}

\def\Q{\mathcal{Q}}
\def\S{\mathcal{S}}
\def\T{\mathcal{T}}

\def\O{\mathcal{O}}

\def\ve{{\varepsilon}}

\newcommand{\suppress}[1]{}

\newcommand {\br} [1] {\ensuremath{ \left( #1 \right) }}
\newcommand {\Br} [1] {\ensuremath{ \left[ #1 \right] }}
\newcommand {\set} [1] {\ensuremath{ \left\lbrace #1 \right\rbrace }}
\newcommand {\minusspace} {\: \! \!}
\newcommand {\smallspace} {\: \!}
\newcommand {\fn} [2] {\ensuremath{ #1 \minusspace \br{ #2 } }}
\newcommand {\Fn} [2] {\ensuremath{ #1 \minusspace \Br{ #2 } }}

\newcommand {\fndec} [3] {\ensuremath{ #1 : \, #2 \rightarrow #3 }}
\newcommand {\mutinf} [2] {\fn{\mathrm{I}}{#1 \smallspace ; \smallspace #2}}
\newcommand {\condmutinf} [3] {\mutinf{#1}{#2 \smallspace \middle\vert \smallspace #3}}
\newcommand {\prob} [1] {\Fn{\Pr}{#1}}
\newcommand {\pr} [2] {\Fn{\Pr_{#1}}{#2}}

\newcommand {\abs} [1] {\ensuremath{ \left| #1 \right| }}
\newcommand {\norm} [1] {\ensuremath{ \left\| #1 \right\| }}
\newcommand {\normsub} [2] {\ensuremath{ \norm{#1}_{#2} }}
\newcommand {\onenorm} [1] {\normsub{#1}{1}}

\newcommand {\relent} [2] {\fn{\mathrm{S}}{#1 \middle\| #2}}
\newcommand {\rminent} [2] {\fn{\mathrm{S}_{\infty}}{#1 \middle\| #2}}
\DeclareMathOperator*{\bigE}{\mathbb{E}}
\newcommand {\expec} [2] {\Fn{\bigE_{\substack{#1}}}{#2}}
\newcommand {\email} [1] {E-mail: \texttt{#1}.}
\title {\textbf{\PdfTitle}}

\author{
    \Rahul\thanks{\CQTCS \email{rahul@comp.nus.edu.sg}} \\
National U. Singapore
    \and
    \Attila\thanks{\CQT \email{attila.pereszlenyi@gmail.com}} \\
National U. Singapore
    \and
    \Penghui\thanks{\CQT \email{pyao@nus.edu.sg}} \\
National U. Singapore }

\date {\today}

\begin{document}
\maketitle

\abstract{A strong direct product theorem for a problem in a given
model of computation states that, in order  to compute $k$ instances
of the problem, if we provide resource which is less than $k$ times
the resource required for computing one instance of the problem with
constant success probability, then the probability of correctly
computing all the $k$ instances together, is exponentially small in
$k$. In this paper, we consider the model of two-party bounded-round
public-coin randomized communication complexity. For a relation $f
\subseteq \X \times \Y \times \Z$ ($\X,\Y,\Z$ are finite sets), let
$\mathrm{R}^{(t), \mathrm{pub}}_\ve (f)$ denote the two-party
$t$-message public-coin communication complexity of $f$ with worst
case error $\ve$. We show that for any relation $f$ and integer $k \geq 1$
$$  \mathrm{R}^{(t), \mathrm{pub}}_{1 -
2^{-\Omega(k/t^2)}}(f^k) = \Omega\(\frac{k}{t} \cdot
\(\mathrm{R}^{(t), \mathrm{pub}}_{1/3}(f) - \O(t^2)\) \) .$$ 
In particular, it implies
a strong direct product theorem for the two-party constant-message
public-coin randomized communication complexity of all relations $f$.
% for which $\mathrm{R}^{(t), \mathrm{pub}}_{1/3}(f)  = \Omega(t^2)$.

Our result for example implies a strong direct product theorem for
the  pointer chasing problem. This problem has been well studied for
understanding round v/s communication trade-offs in both classical
and quantum communication
protocols~\cite{Nisan:1991:RCC:103418.103463,Klauck:2000:QPC:335305.335396,Ponzio:2001:CCP:374952.374989,Klauck:2001:IQC:380752.380786,Jain:2002:QCC:646840.708805}.

We show our result using information theoretic arguments.  Our
arguments and techniques build on the ones used in
Jain~\cite{Jain:2011}, where a strong direct product theorem for the
two-party one-way public-coin communication complexity of all
relations is shown (that is the special case of our result when
$t=1$). One key tool used in our work and also in
Jain~\cite{Jain:2011} is a  message compression technique due to
Braverman and Rao~\cite{Braverman2011}, who used it to show a direct
sum theorem for the two-party bounded-round public-coin randomized
communication complexity of all relations. Another important  tool
that we use is a correlated sampling protocol,  which for example,
has been used in Holenstein~\cite{Holenstein2007} for proving a
parallel repetition theorem  for two-prover games.

}

\section{Introduction}
A fundamental question in complexity theory is how much resource is
needed to solve $k$ independent instances of a problem compared to
the resource required to solve one instance. More specifically,
suppose for solving one instance of a problem with probability of
correctness $p$, we require $c$ units of some resource in a given
model of computation. A natural way to solve $k$ independent
instances of the same problem is to solve them independently, which
needs $k \cdot c$ units of resource and the overall success
probability is $p^k$. A {\em strong direct product} theorem for this
problem would state that any algorithm, which solves $k$ independent
instances of this problem with $o(k\cdot c)$ units of the resource,
can only compute all the $k$ instances correctly with probability at
most $p^{-\Omega(k)}$.

In this work, we are concerned with the model of communication
complexity which was introduced by Yao \cite{Yao:1979}.  In this
model there are different parties who wish to compute a joint
relation of their inputs. They do local computation, use
public/private coins, and communicate between them to achieve
this task. The resource that is counted is the number of bits
communicated. The text by Kushilevitz and
Nisan~\cite{Kushilevitz96} is an excellent reference for this model.
Direct product questions and the weaker {\em direct sum} questions
have been extensively investigated in different sub-models of
communication complexity.  A direct sum theorem states that in order
to compute $k$ independent instances of a problem, if we provide
resources less than $k$ times the resource required to compute one
instance of the problem with the constant success probability $p <
1$, then the success probability for computing all the $k$ instances
correctly is at most a constant $q < 1$. Some examples of known
direct product theorems are: Parnafes, Raz and
Wigderson's~\cite{Parnafes:1997:DPR:258533.258620} theorem for {\em
forests} of communication protocols;
Shaltiel's~\cite{Shaltiel:2004:TPS:1018404.1018405} theorem for the
{\em discrepancy bound} (which is  a lower bound on the {\em
distributional} communication complexity) under the uniform
distribution; extended to arbitrary distributions
by Lee, Shraibman and \v{S}palek~\cite{Lee2008}; extended  to the multiparty case by Viola and Wigderson~\cite{ViolaW08}; extended to the generalized discrepancy bound by Sherstov~\cite{Sherstov:2011:SDP:1993636.1993643}; 
Jain, Klauck and Nayak's~\cite{Jain2008} theorem for subdistribution bound; Klauck, \v{S}palek, de Wolf's~\cite{KlauckSdW04} theorem for the quantum communication complexity of the {\em set disjointness} problem; Klauck's \cite{Klauck2010} theorem for the public-coin communication complexity of the set-disjointness problem
(which was re-proven using very different arguments in
Jain~\cite{Jain:2011}); Ben-Aroya, Regev, and de Wolf's~\cite{Ben-AroyaRdW08} theorem for the one-way quantum communication complexity of the {\em index} function problem;  Jain's \cite{Jain:2011} theorem for
randomized one-way communication complexity and
Jain's~\cite{Jain:2011} theorem for {\em conditional relative
min-entropy bound} (which is a lower bound on the public-coin
communication complexity). Direct sum theorems have been shown in
the public-coin one-way model~\cite{Jain:2003:DST:1759210.1759242},
public-coin simultaneous message passing
model~\cite{Jain:2003:DST:1759210.1759242}, entanglement-assisted
quantum one-way communication
model~\cite{Jain:2005:PEM:1068502.1068658}, private-coin
simultaneous message passing
model~\cite{Jain:2009:NRS:1602931.1603157} and constant-round
public-coin two-way model~\cite{Braverman2011}. On the other hand,
strong direct product conjectures have been shown to be false by
Shaltiel~\cite{Shaltiel:2004:TPS:1018404.1018405} in some models of
distributional communication complexity (and of {\em query
complexity} and {\em circuit depth complexity}) under specific
choices for the error parameter.

Examples of direct product theorems in others models of computation
include Yao's {\em XOR lemma}~\cite{Yao82}, Raz's
\cite{Raz:1995:PRT:225058.225181} theorem for two-prover games;
Shaltiel's~\cite{Shaltiel:2004:TPS:1018404.1018405} theorem for {\em
fair decision trees}; Nisan, Rudich and
Saks'~\cite{Nisan:1999:PHB:305673.305747} theorem for {\em decision
forests}; Drucker's~\cite{Drucker:2011:IDP:2013878.2014045} theorem
for randomized query complexity;
Sherstov's~\cite{Sherstov:2011:SDP:1993636.1993643} theorem for {\em
approximated polynomial degree} and Lee and Roland's~\cite{Lee:2011}
theorem for quantum query complexity.  Besides their inherent
importance, direct product theorems have had various important applications
such as in {\em Probabilistically checkable
proofs}~\cite{Raz:1995:PRT:225058.225181}; in circuit
complexity~\cite{Yao82} and in showing time-space tradeoffs~\cite{KlauckSplaekdewolf:2004:QCS:1032645.1033156,Ambainisspalekdewolf09,Klauck2010}.

In this paper, we show a direct product theorem for the two-party
bounded-round public-coin randomized communication complexity. In
this model, for computing a relation $f\subseteq \X\times\Y\times\Z$
($\X,\Y,\Z$ are finite sets), one party say Alice, is given an input
$x\in\X$ and the other party say Bob, is given an input $y\in Y$.
They are supposed to do local computations using public-coins shared
between them, communicate a fixed number of messages between
them and at the end, output an element $z\in Z$. They are said
to succeed if $(x,y,z)\in f$.  For a natural number $t \geq 1$ and
$\ve \in (0,1)$, let $\mathrm{R}^{(t), \mathrm{pub}}_\ve (f)$ denote
the two-party $t$-message public-coin communication complexity of
$f$ with worst case error $\ve$, that is the communication of the
best public-coin protocol between Alice and Bob with $t$ messages
exchanged between them, and the error (over the public coins) on any
input $(x,y)$ being at most $\ve$. We show the following.
\begin{theorem}\label{thm:direct product}
    Let  $\X$, $\Y$, $\Z$ be finite sets,
    $f\subseteq\X\times\Y\times\Z$ a relation, $\ve > 0$ and $k, t \geq1$ be integers. There exists a constant $\kappa$ such that,
    \[ \mathrm{R}^{(t),\mathrm{pub}}_{1-(1-\ve/2)^{\Omega(k \ve^2/t^2)}}(f^k) = \Omega \( \frac{\ve \cdot k}{t} \cdot \(\mathrm{R}^{(t),\mathrm{pub}}_{\ve}(f) - \frac{\kappa t^2}{\ve^2} \) \).\]
\end{theorem}
In particular, it implies a strong direct product theorem for the
two-party constant-message public-coin randomized communication
complexity of all relations $f$\footnote{When $\mathrm {R}^{(t),\mathrm{pub}}_{\ve}(f)$ is a constant, then a direct product result can be shown via direct arguments as for example in~\cite{Jain:2011,Sherstov:2011:SDP:1993636.1993643}.}. 
% for which $\mathrm{R}^{(t), \mathrm{pub}}_{1/3}(f)  = \Omega(t^2)$.
Our result generalizes the result of Jain \cite{Jain:2011} which can be regarded as the special case when $t=1$.

As a direct consequence of our result we get a direct product
theorem for the {\em pointer chasing} problem defined as follows.
Let $n, t\geq1$ be integers. Alice and Bob are given functions $F_A:
[n]\rightarrow [n]$ and $F_B: [n]\rightarrow [n]$, respectively.
Let $F^t$ represent alternate composition of $F_A$ and $F_B$ done
$t$ times, starting with $F_A$. The parties are supposed to
communicate and determine $F^t(1)$. In the bit version of the
problem, the players are supposed to output the least significant
bit of $F^t(s)$. We refer to the $t$-pointer chasing problem as
$\mathrm{FP}_t$ and the bit version as $\mathrm{BP}_t$. The pointer
chasing problem naturally captures the trade-off between number of
messages exchanged and the communication used. There is a
straightforward  $t$-message deterministic protocol with $t\cdot\log
n$ bits of communication for both $\mathrm{FP}_t$ and
$\mathrm{BP}_t$. However if only $t-1$ messages are allowed to be
exchanged between the parties, exponentially more communication is
required.  The communication complexity of this problem has been
very well studied both in the classical and quantum models of
communication complexity
~\cite{Nisan:1991:RCC:103418.103463,Klauck:2000:QPC:335305.335396,Ponzio:2001:CCP:374952.374989,Klauck:2001:IQC:380752.380786,Jain:2002:QCC:646840.708805}.
The best lower bounds we know so far are as follows (below $\mathrm{Q}^{(t)}(\cdot)$  stands for the $t$-message quantum communication complexity).
\begin{theorem}\label{thm:pointerchasing}
For  integer $t \geq 1$,
\begin{enumerate}
\item~\cite{Ponzio:2001:CCP:374952.374989} $\mathrm{R}^{(t-1),\mathrm{pub}}_{1/3}(\mathrm{FP}_t)\geq\Omega(n\log^{(t-1)}n).$

\item~\cite{Ponzio:2001:CCP:374952.374989} $\mathrm{R}^{(t-1),\mathrm{pub}}_{1/3}(\mathrm{BP}_t)\geq\Omega(n).$

\item~\cite{Jain:2002:QCC:646840.708805} $\mathrm{Q}^{(t-1)}_{1/3}(\mathrm{FP}_t)\geq\Omega(n\log^{(t-1)}n).$

\end{enumerate}
\end{theorem}
As a consequence of Theorem~\ref{thm:direct product} we get  strong
direct product results for this problem. Note that in the descriptions
of $\mathrm{FP}_t$ and $\mathrm{BP}_t$, $t$ is a fixed constant, not
dependent on the input size.
\begin{cor}\label{cor:pointerchasingsdpt}
For  integers $t,k\geq 1$,
\begin{enumerate}
\item
$\mathrm{R}^{(t-1),\mathrm{pub}}_{1-2^{-\Omega(k/t^2)}}(\mathrm{FP}_t^k)\geq\Omega\br{\frac{k}{t}\cdot
n\log^{(t-1)}n}$.
\item $\mathrm{R}^{(t-1),\mathrm{pub}}_{1-2^{-\Omega(k/t^2)}}(\mathrm{BP}_t^k)\geq\Omega\br{\frac{k}{t}\cdot n}$.
\end{enumerate}
\end{cor}

\subsection*{Our techniques}
We prove our direct product result using information theoretic
arguments. Information theory is a versatile tool in communication
complexity, especially in proving lower bounds and direct sum and
direct product
theorems~\cite{Chakrabarti:2001:ICD:874063.875561,Bar-Yossef2002a,Jain:2003:DST:1759210.1759242,Jain:2003:LBB:946243.946331,Jain:2005:PEM:1068502.1068658,Jain:2009:NRS:1602931.1603157,Barak:2010:CIC:1806689.1806701,Braverman2011,Jain:2011}.
The broad argument that we use is as follows. For a given
relation $f$, let  the communication required for computing one
instance with $t$ messages and constant success  be $c$. Let us
consider a  protocol for computing $f^k$ with $t$ messages and
communication cost $o(kc)$.  Let us condition on success on some $l$
coordinates.  If the overall success in these $l$ coordinates is
already as small as we want then we are done and stop. Otherwise we exhibit another coordinate $j$ outside of these $l$
coordinates such that the success in the $j$-th coordinate, even
conditioned on the success in the $l$ coordinates, is bounded away
from $1$. This way the overall success keeps going down and becomes
exponentially small eventually. We do this argument in the
distributional setting where one is concerned with average error
over the inputs coming from a specified distribution rather than the
worst case error over all inputs. The distributional setting can
then be related to the worst case setting by the well known Yao's
principle~\cite{Yao:1979}.

More concretely, let $\mu$ be a distribution on $\X \times \Y$,
possibly non-product across $\X$ and $\Y$. Let $c$ be the minimum
communication required for computing $f$ with $t$-message protocols
having error at most $\ve$ averaged over $\mu$. Let us consider the
inputs for $f^k$ drawn from the distribution $\mu^{ k}$ ($k$
independent copies of $\mu$). Consider a $t$-message protocol
$\mathcal{P}$ for $f^k$ with communication $o(kc)$ and for the rest
of the argument condition on success on a set $C$ of coordinates. If
the success probability of this event is as small as we desire then
we are done. Otherwise we exhibit a new coordinate $j\notin C$
satisfying the following conditions: first the distribution of
inputs $X_j Y_j$ (of Alice and Bob respectively) in the $j$-th
coordinate is quite close to $\mu$; second the joint distribution
$X_jY_jM$ (where $M$ is the message transcript of $\mathcal{P}$) can be
approximated very well by Alice and Bob using a $t$ message protocol
for $f$, when they are given input according to $\mu$, using
communication less than $c$. This shows that success in the $j$-th
coordinate must be bounded away from one. Since we can simulate each
message only approximately, in order to keep the overall error
bounded, we are able to make our argument for protocols with a
bounded number of message exchanges.

One difficulty that is faced in this argument is that
since $\mu$ may be a non-product distribution, Alice and Bob may
obtain information about each other's input in the $j$-th coordinate
via their inputs in other coordinates.  This  is overcome by
splitting the distribution $\mu$ into a convex combination of
several product distributions. This idea of splitting a non-product
distribution into convex combination of product distributions has
been used in several previous works to handle non-product
distributions in different
settings~\cite{Razborov92,Raz:1995:PRT:225058.225181,Bar-Yossef2002a,Holenstein2007,Barak:2010:CIC:1806689.1806701,Braverman2011,Jain:2011}.
Some important tools that we use in our arguments are a message
compression protocol due to Braverman and Rao~\cite{Braverman2011}
and the {\em correlated sampling} protocol that appeared for example
in Holenstein~\cite{Holenstein2007}.

\subsubsection*{Organization} The rest of the paper is organized as follows. In Section~\ref{sec:Preliminaries}, we present some background on information theory and communication complexity.  In Section~\ref{sec:main theorem}, we prove our main result Theorem~\ref{thm:direct product}, starting with some lemmas that are helpful in building the proof.
%In Section~\ref{sec:conclusion}, we make some concluding remarks and list some open problems.

\section{Preliminaries}
\label{sec:Preliminaries}

\subsection*{Information theory}

For integer $n \geq 1$, let $[n]$ represent the set $\{1,2, \ldots,
n\}$. Let $\X$, $\Y$ be finite sets and $k$ be a natural number. Let
$\X^k$ be the set $\X\times\cdots\times\X$, the cross product of
$\X$ $k$ times. Let $\mu$ be a (probability) distribution on $\X$.
Let $\mu(x)$ represent the probability of $x\in\X$ according to
$\mu$. Let $X$ be a random variable distributed according to $\mu$,
which we denote by $X\sim\mu$. We use the same symbol to represent
a random variable and its distribution whenever it is clear from
the context.  The expectation value of some function $f$ on $\X$ is
denoted as
\[ \expec{x \leftarrow X}{f(x)} \defeq\sum_{x\in\X} \prob{X=x} \cdot f(x). \]
The entropy of $X$ is defined to be
$\mathrm{H}(X)\defeq-\sum_x\mu(x) \cdot \log\mu(x)$. For two distributions
$\mu$, $\lambda$ on $\X$, the distribution $\mu \otimes \lambda$ is
defined as $(\mu\otimes\lambda)(x_1,x_2)\defeq\mu(x_1)\cdot\lambda(x_2)$.
Let $\mu^k\defeq\mu\otimes\cdots\otimes\mu$, $k$ times. The $\ell_1$
distance between $\mu$ and $\lambda$ is defined to be half of the
$\ell_1$ norm of $\mu - \lambda$; that is
$$\|\lambda-\mu\|_1\defeq\frac{1}{2}\sum_x|\lambda(x)-\mu(x)|=\max_{S\subseteq\X}|\lambda_S-\mu_S| , $$
where $\lambda_S \defeq\sum_{x\in S}\lambda(x)$. We say that
$\lambda$ is $\ve$-close to $\mu$ if $\|\lambda-\mu\|_1\leq\ve$. The
relative entropy between distributions $X$ and $Y$ on $\X$ is
defined as
\[  \relent{X}{Y} \defeq \sum_{x\in\X}
    \prob{X=x} \cdot \log \frac{\prob{X=x}}{\prob{Y=x}} .\]
The relative min-entropy between them is defined as
\[ \rminent{X}{Y} \defeq \max_{x\in\X}
    \set{ \log \frac{\prob{X=x}}{\prob{Y=x}} }.\]
It is easy to see that $\relent{X}{Y} \leq \rminent{X}{Y}$. Let
$X,Y,Z$ be jointly distributed random variables. Let $Y_x$ be the
distribution of $Y$ conditioned on $X=x$. The conditional entropy of
$Y$ conditioned on $X$ is defined as $\mathrm{H}(Y|X) \defeq
\expec{x\leftarrow X}{\mathrm{H}(Y_x)} =
\mathrm{H}(XY)-\mathrm{H}(X)$. The mutual information between $X$
and $Y$ is defined as
\[  \mutinf{X}{Y} \defeq \mathrm{H}(X)+\mathrm{H}(Y)-\mathrm{H}(XY)
    = \expec{y \leftarrow Y}{\relent{X_y}{X}}
    = \expec{x \leftarrow X}{\relent{Y_x}{Y}}. \]
It is easily seen that $\mutinf{X}{Y} = \relent{XY}{X \otimes Y}$.
We say that $X$ and $Y$ are independent iff $\mutinf{X}{Y} = 0$. The
conditional mutual information between $X$ and $Y$, conditioned on
$Z$, is defined as
\[  \condmutinf{X}{Y}{Z} \defeq
    \expec{z \leftarrow Z}{\condmutinf{X}{Y}{Z=z}}
    = \mathrm{H}\br{X|Z}+\mathrm{H}\br{Y|Z}-\mathrm{H}\br{XY|Z} .\]
The following {\em chain rule} for mutual information is easily
seen,
$$\mutinf{X}{YZ} = \mutinf{X}{Z} + \condmutinf{X}{Y}{Z} .$$  Let $X,X', Y, Z $ be jointly distributed random variables. We define the joint
distribution of $(X'Z)(Y|X)$ by
$$\Pr[(X'Z)(Y|X)=x,z,y]
\defeq \Pr[X'=x, Z=z] \cdot \Pr[Y=y|X=x].$$
We say that $X$, $Y$, $Z$ is a Markov chain iff $XYZ=(XY)(Z|Y)$ and
we denote it by $X\leftrightarrow Y\leftrightarrow Z$. It is easy to
see that $X$, $Y$, $Z$ is a Markov chain if and only if
$\condmutinf{X}{Z}{Y}=0$. Ibinson, Linden and Winter~\cite{Ben2008}
showed that if $\condmutinf{X}{Y}{Z}$ is small then $XYZ$ is close
to being a Markov chain.
\begin{lemma}[\cite{Ben2008}]
    \label{lem:mutual inf and relative ent}
    For any random variables $X$, $Y$ and $Z$, it holds that
    \[ \condmutinf{X}{Z}{Y} = \min \set{ \relent{XYZ}{X'Y'Z'} :
    X' \leftrightarrow Y'\leftrightarrow Z'}. \]
    The minimum is achieved by distribution
    $X'Y'Z'=(XY)(Z|Y)$.
\end{lemma}

We will need the following basic facts. A very good text for
reference on information theory is~\cite{CoverT91}.

\begin{fact}
\label{fact:relative entropy joint convexity} Relative entropy is
jointly convex in its arguments. That is, for distributions $\mu,
\mu^1, \lambda, \lambda^1 \in \X$,
$$ \relent{p \mu  + (1-p) \mu^1}{\lambda + (1-p) \lambda^1} \leq p \cdot \relent{\mu}{\lambda} + (1-p) \cdot \relent{\mu^1}{\lambda^1} .$$
\end{fact}
\begin{fact}
    \label{fact:relative entropy splitting}
Relative entropy satisfies the following chain rule. Let $XY$ and
$X^1Y^1$ be random variables on $\X\times\Y$.
    It holds that
    \[ \relent{X^1Y^1}{XY} = \relent{X^1}{X}
    + \expec{x\leftarrow X^1} {\relent{Y^1_x}{Y_x}}.\]
    In particular, using Fact~\ref{fact:relative entropy joint convexity}
    \[ \relent{X^1Y^1}{X\otimes Y}
    = \relent{X^1}{X} + \expec{x\leftarrow X^1}{\relent{Y^1_x}{Y}}
    \geq \relent{X^1}{X} + \relent{Y^1}{Y}.\]
\end{fact}

\begin{fact} \label{fact:mutinf is min}
    Let $XY$ and $X^1Y^1$ be random variables on $\X\times\Y$.
    It holds that
    \[  \relent{X^1Y^1}{X\otimes Y}
    \geq \relent{X^1Y^1}{X^1\otimes Y^1}=\mutinf{X^1}{Y^1}. \]
\end{fact}

\begin{fact}
    \label{fact:one norm and rel ent}
    For distributions $\lambda$ and $\mu$,
    \[ 0 \leq \onenorm{\lambda-\mu} \leq \sqrt{\relent{\lambda}{\mu}}. \]
\end{fact}

\begin{fact}
    \label{fact:-1 inequality}
    Let $\lambda$ and $\mu$ be distributions on $\X$.
    For any subset $\S \subseteq \X$,
    it holds that
    \[ \sum_{x \in \S} \lambda(x) \cdot
    \log \frac{\lambda(x)}{\mu(x)} \geq -1 .\]
\end{fact}

\begin{fact}
\label{fact:subsystem monotone} The $\ell_1$ distance and relative
entropy are monotone non-increasing when subsystems are considered.
Let $X,Y,X^1,Y^1$ be random variables, then
$$ \onenorm{XY - X^1Y^1}  \geq \onenorm{X - X^1} \quad \mbox{and} \quad \relent{XY}{X^1Y^1} \geq \relent{X}{X^1} .$$
\end{fact}

\begin{fact}
    \label{fact:l1 monotone}
    For  function \fndec{f}{\X\times \R}{\Y } and random variables
    $X, Y$ on $\X$ and $R$ on $\R$, such that $R$ is independent of $(XY)$, it holds that
    \[ \onenorm{Xf(X,R) - Yf(Y,R)} = \onenorm{X-Y}. \]
\end{fact}
The following definition was introduced by
Holenstein~\cite{Holenstein2007}. It plays a critical role in his
proof of a parallel repetition theorem for two-prover games.
\begin{definition}[\cite{Holenstein2007}]\label{def:embeddable}
For two distributions $(X_0Y_0)$ and $(X_1SY_1T)$, we say that
$(X_0,Y_0)$ is \br{1-\ve}-embeddable in $(X_1S,Y_1T)$ if there
exists a probability distribution $R$ over a set $\R$, which is
independent of $X_0Y_0$ and functions $f_A:\X\times\R\rightarrow\S$,
$f_B:\Y\times\R\rightarrow\T$, such that
\[  \onenorm{X_0Y_0f_A(X_0,R)f_B(Y_0,R) - X_1Y_1ST}
    \leq\ve. \]
\end{definition}
The following lemma was shown by Holenstein~\cite{Holenstein2007}
using a correlated sampling protocol.
\begin{lemma}[\cite{Holenstein2007}]
    \label{lem:distance embeddable}
    For random variables $S$, $X$ and $Y$, if
    \[ \onenorm{SXY-(XY)(S|X)} \leq \ve \] and
    \[ \onenorm{SXY-(XY)(S|Y)} \leq \ve, \]
    then $(X,Y)$ is \br{1-4\ve}-embeddable in $(XS,YS)$.
\end{lemma}
We will need the following generalization of the previous lemma.
\begin{lemma}\label{lem:relent embeddable}
    For  joint random variables $(A',B',C')$ and $(A,B)$, satisfying
    \begin{align*}
       \relent{A'B'}{AB} &\leq \ve \\
        \expec{(a,c)\leftarrow A',C'} {\relent{B'_{a,c}}{B_a}} &\leq \ve \qquad \mbox{and} \\
        \expec{(b,c)\leftarrow B',C'} {\relent{A'_{b,c}}{A_b}} &\leq \ve,
    \end{align*}
    it holds that $(A,B)$ is \br{1-5\sqrt{\ve}}-embeddable in $(A'C',B'C')$.
\end{lemma}
\begin{proof}
    Using the definition of the relative entropy, we have
    the following.
    \begin{align*}
        \expec{(a,c)\leftarrow A',C'} {\relent{B'_{a,c}}{B_a}}
        - \expec{(a,c)\leftarrow A',C'} {\relent{B'_{a,c}}{B'_a}}
        &= \expec{(a,b,c) \leftarrow A',B',C'}
        {\log \frac{\prob{B'=b|A'=a}}{\prob{B=b|A=a}}} \\
        &= \expec{a \leftarrow A'}{\relent{B'_a}{B_a}} \quad \geq \quad 0 .
    \end{align*}
    This means that
    \begin{align}
        \expec{(a,c)\leftarrow A',C'} {\relent{B'_{a,c}}{B'_a}}
        \leq \expec{(a,c)\leftarrow A',C'} {\relent{B'_{a,c}}{B_a}}
        \leq \ve.
        \label{eqn:e0}
    \end{align}
    Then
    \begin{align}
        \expec{(a,c)\leftarrow A',C'} {\relent{B'_{a,c}}{B_a'}}
        &= \relent{A'B'C'}{\br{A'C'}\br{B'|A'}}
        \label{eqn:e1} \\
        &= \relent{A'B'C'}{\br{A'B'}\br{C'|A'}}
        \label{eqn:e2} \\
        &\geq \onenorm{A'B'C' - \br{A'B'}\br{C'|A'}}^2 .
        \label{eqn:e3}
    \end{align}
    Above, Eq.~\eqref{eqn:e1} follows from the definition of the relative
    entropy, Eq.~\eqref{eqn:e2} follows because $\br{A'C'}\br{B'|A'}$
    and $\br{A'B'}\br{C'|A'}$ are identically distributed,     and Eq.~\eqref{eqn:e3} follows from Fact~\ref{fact:one norm and rel ent}.
    Now from Equations \eqref{eqn:e3} and \eqref{eqn:e0} we get
    \begin{align*}
        \onenorm{A'B'C' - \br{A'B'}\br{C'|A'}} &\leq \sqrt{\varepsilon}.
        \intertext{By similar arguments we get}
        \onenorm{A'B'C' - \br{A'B'}\br{C'|B'}} &\leq \sqrt{\varepsilon}.
    \end{align*}
The inequalities above and  Lemma~\ref{lem:distance embeddable} imply
that $(A',B')$ is $\br{1-4\sqrt{\ve}}$-embeddable in $(A'C',B'C')$.
Furthermore from Fact~\ref{fact:one norm and rel ent} and
$\relent{A'B'}{AB} \leq \ve $ we get
$$ \onenorm{A'B' - AB} \leq \sqrt{\ve} . $$
Finally using the inequality above and Fact~\ref{fact:l1 monotone} we get that
$(A,B)$ is $\br{1-5\sqrt{\ve}}$-embeddable in $(A'C',B'C')$.
\end{proof}

\subsection*{Communication complexity}
\label{sec:Communication complexity} Let $f \subseteq \X \times \Y
\times \Z$ be a relation, $t \geq 1$ be an integer and $\ve \in
(0,1)$. In this work we only consider {\em complete} relations, that
is for every $(x,y) \in \X \times \Y$, there is some $z \in \Z$ such
that $(x,y,z) \in f$. In the two-party $t$-message public-coin model
of communication, Alice with input $x \in \X$ and Bob with input $y
\in \Y$, do local computation using public coins shared between them
and exchange $t$ messages, with Alice sending the first
message. At the end of their protocol the party receiving the $t$-th
message outputs some $z \in \Z$. The output is declared correct if
$(x,y,z) \in f$ and wrong otherwise. Let
$\mathrm{R}^{(t),\mathrm{pub}}_{\ve}(f)$ represent the two-party
$t$-message public-coin communication complexity of $f$ with worst
case error $\ve$, i.e., the communication of the best two-party
$t$-message public-coin protocol for $f$ with error for each input
$(x,y)$ being at most~$\ve$. We similarly consider two-party
$t$-message deterministic protocols where there are no public coins
used by Alice and Bob. Let $\mu \in \X \times \Y$ be a distribution.
We let $\mathrm{D}_{\ve}^{(t),\mu}(f)$ represent the two-party
$t$-message distributional communication complexity of $f$ under
$\mu$ with expected error $\ve$, i.e., the communication of the best
two-party $t$-message deterministic protocol for $f$, with
distributional error  (average error over  the inputs) at most $\ve$
under $\mu$.    Following is a consequence of the  min-max theorem
in game theory, see e.g.,~\cite[Theorem~3.20,
page~36]{Kushilevitz96}.
\begin{lemma}[Yao's principle, \cite{Yao:1979}]
\label{lem:yaos principle}
$\mathrm{R}^{(t),\mathrm{pub}}_{\ve}(f)=\max_{\mu}\mathrm{D}^{(t),\mu}_{\ve}(f)$.
\end{lemma}
The following fact about communication protocols can be verified easily.
\begin{fact}
\label{fact:commindep} Let there be $t$ messages $M_1, \ldots, M_t$
in a deterministic communication protocol between Alice and Bob with
inputs $X,Y$ respectively where $X$ and $Y$ are independent. Then
for any $s \in [t]$, $X$ and $Y$ are independent even conditioned on
$M_1, \ldots, M_s$.
\end{fact}

\section{Proof of Theorem~\ref{thm:direct product}}
\label{sec:main theorem} We start by showing a few lemmas which
are helpful in the proof of the main result. The following lemma was
shown by Jain~\cite{Jain:2011} and follows primarily from a  message
compression argument due to Braverman and Rao~\cite{Braverman2011}.
\begin{theorem}[\cite{Braverman2011,Jain:2011}]
    \label{thm:braverman rao protocol}
    Let $\delta >0, c \geq 0$. Let $X', Y', N$ be random variables for which
    $Y' \leftrightarrow X' \leftrightarrow N$ is a Markov chain and the following holds,
    \[ \pr{(x,y,m)\leftarrow X',Y',N}
    {\log\frac{\prob{N=m|X'=x}}{\prob{N=m|Y'=y}}>c}
    \leq\delta . \]
    There exists a public-coin protocol between Alice and Bob, with inputs  $X', Y'$ respectively, with a single message from Alice to Bob of  $c+\O(\log(1/\delta))$ bits, such that at the end of  the protocol, Alice and Bob both possess a random variable $M$ satisfying $ \onenorm{X'Y'N-X'Y'M} \leq  2\delta $.
\end{theorem}
We will need the following generalization of the above.
\begin{lemma}\label{lem:generalized braverman rao}
    Let  $ c\geq0, 1 > \ve >0, \ve' >0 $. Let $X', Y', M'$ be random variables for which the following holds,
    \[ \condmutinf{X'}{M'}{Y'} \leq c
    \text{ and }
    \condmutinf{Y'}{M'}{X'} \leq \ve . \]
There exists a public-coin  protocol between Alice and Bob, with
inputs $X',Y'$ respectively, with a single message from Alice to Bob
of   $\frac{c+5}{\ve'}+\O(\log\frac{1}{\ve'})$ bits,
such that at the end of the protocol,
    Alice and Bob both possess a random variable $M$  satisfying $ \onenorm{X'Y'M'-X'Y'M} \leq 3\sqrt{\varepsilon}+6\ve' $.
\end{lemma}
\begin{proof}
Let us introduce a new random variable $N$ with joint distribution
$X'Y'N \defeq (X'Y')(M'|X')$. Note that $Y' \leftrightarrow X'
\leftrightarrow N$ is a Markov chain. Using Lemma~\ref{lem:mutual
inf and relative ent}, we have
\[ \relent{X'Y'M'}{X'Y'N} = \condmutinf{Y'}{M'}{X'} \leq \ve .\]
Applying Fact~\ref{fact:one norm and rel ent}, we get that
$\onenorm{X'Y'M'-X'Y'N} \leq \sqrt{\ve}$. Using this, the following
claim, and Theorem~\ref{thm:braverman rao protocol} we conclude the
desired. 
\end{proof}
\begin{claim}
\[  \pr{(m,x,y)\leftarrow M,X',Y'}
    {\log \frac{\prob{N=m|X'=x}}{\prob{N=m|Y'=y}}
    \geq \frac{c+5}{\ve'}}
    \leq 3 \ve'+\sqrt{\ve}.\]
\end{claim}
\begin{proof}
For any $m$, $x$, $y$ it holds that
\begin{align}
    \log \frac{\prob{N=m|X'=x}}{\prob{N=m|Y'=y}}
    &= \log \frac{\prob{N=m|X'=x, Y'=y}}{\prob{N=m|Y'=y}} \nonumber \\
    &= \log \frac{\prob{N=m|X'=x, Y'=y}}{\prob{M'=m|X'=x, Y'=y}} +
    \log \frac{\prob{M'=m|X'=x, Y'=y}}{\prob{M'=m|Y'=y}} \nonumber \\
    &\qquad {} + \log \frac{\prob{M'=m, Y'=y}}{\prob{N=m, Y'=y}}.
\label{eqn:eq6}
\end{align}
We bound each term above separately. For the first one, let us
define the set
\[  G_1 \defeq \set{(m,x,y) :
    \log \frac{\prob{N=m|X'=x, Y'=y}}
    {\prob{M'=m|X'=x, Y'=y}} \leq \frac{\ve+1}{\ve'}} .\]
Consider,
\begin{align}
    0 &\geq -\expec{(x,y)\leftarrow X',Y'} {\relent{M'_{xy}}{N_{xy}}} \nonumber \\
    &= \expec{(m,x,y)\leftarrow M',X',Y'}
    {\log \frac{\prob{N=m|X'=x, Y'=y}}{\prob{M'=m|X'=x, Y'=y}}}
    \label{eqn:a1} \\
    &= \sum_{(m,x,y) \in G_1}
    \prob{M'=m, X'=x, Y'=y} \cdot
    \log \frac{\prob{N=m|X'=x, Y'=y}}{\prob{M'=m|X'=x, Y'=y}} \nonumber \\
    &\qquad {} + \sum_{(m,x,y) \notin G_1}
    \prob{M'=m, X'=x, Y'=y} \cdot
    \log \frac{\prob{N=m|X'=x, Y'=y}}{\prob{M'=m|X'=x, Y'=y}} \nonumber \\
    &\geq \sum_{(m,x,y) \in G_1}
    \prob{M'=m, X'=x, Y'=y} \cdot
    \log \frac{\prob{N=m|X'=x, Y'=y}}{\prob{M'=m|X'=x, Y'=y}} \nonumber \\
    &\qquad {} + \prob{\br{M',X',Y'} \notin G_1} \cdot \frac{\ve+1}{\ve'}
    \label{eqn:a2} \\
    &= \sum_{(m,x,y) \notin G_1}
    \prob{M'=m, X'=x, Y'=y} \cdot
    \log \frac{\prob{M'=m|X'=x, Y'=y}}{\prob{N=m|X'=x, Y'=y}} \nonumber \\
    &\qquad {} - \relent{M'X'Y'}{NX'Y'}
    + \prob{\br{M',X',Y'} \notin G_1} \cdot \frac{\ve+1}{\ve'}
    \label{eqn:a3} \\
    &\geq -1 - \varepsilon
    + \prob{\br{M',X',Y'} \notin G_1} \cdot \frac{\ve+1}{\ve'} .
    \label{eqn:a4}
\end{align}
Above, Eq.~\eqref{eqn:a1} and Eq.~\eqref{eqn:a3} follow from the
definition of the relative entropy, and Eq.~\eqref{eqn:a2} follows
from the definition of $G_1$. To get Eq.~~\eqref{eqn:a4}, we use
Fact~\ref{fact:-1 inequality}. Eq.~\eqref{eqn:a4} implies that
$\prob{\br{M',X',Y'} \notin G_1} \leq \ve'$.

To upper bound the second term let  us define
\[  G_2 \defeq \set{(m,x,y) :
    \log \frac{\prob{M'=m|X'=x, Y'=y}}
    {\prob{M'=m|Y'=y}} \leq \frac{c+1}{\ve'}} .\]
Consider,
\begin{align}
    c &\geq \condmutinf{M'}{X'}{Y'} \label{eqn:b1} \\
    &= \expec{(m,x,y)\leftarrow M',X',Y'}
    {\log \frac{\prob{M'=m|X'=x, Y'=y}}{\prob{M'=m|Y'=y}}}
    \label{eqn:b2} \\
    &=\sum_{(m,x,y)\in G_2} \prob{M'=m, X'=x, Y'=y} \cdot
    \log \frac{\prob{M'=m|X'=x,Y'=y}}{\prob{M'=m|Y'=y}} \nonumber\\
    &\qquad {} +\sum_{(m,x,y)\not\in G_2} \prob{M'=m, X'=x, Y'=y} \cdot
    \log \frac{\prob{M'=m|X'=x,Y'=y}}{\prob{M'=m|Y'=y}} \nonumber\\
    &\geq \frac{c+1}{\varepsilon'} \cdot
    \prob{\br{M',X',Y'} \notin G_2} - 1  .
    \label{eqn:b3}
\end{align}
Above Eq.~\eqref{eqn:b1} is one of the assumptions in the lemma;
Eq.~\eqref{eqn:b2} follows from the definition of the conditional
mutual information; Eq.~\eqref{eqn:b3} follows from the definition
of $G_2$ and Fact \ref{fact:-1 inequality}. Eq.~\eqref{eqn:b3}
implies that $\prob{\br{M',X',Y'} \notin G_2} \leq \ve'$.

To bound the last term define
\[  G_3 \defeq \set{(m,x,y) :
    \log \frac{\prob{M'=m, Y'=y}}{\prob{N=m, Y'=y}}
    \leq \frac{\ve+1}{\ve'}} .\]
Consider,
\begin{align}
    \varepsilon &\geq \relent{X'Y'M'}{X'Y'N} \nonumber\\
    &\geq \relent{Y'M'}{Y'N} \label{eqn:myc1} \\
    &= \expec{(m,x,y) \leftarrow M',X',Y'}
    {\log \frac{\prob{M'=m, Y'=y}}{\prob{N=m, Y'=y}}} \nonumber \\
    &=\sum_{(m,x,y)\in G_3} \prob{M'=m, X'=x, Y'=y} \cdot
    \log \frac{\prob{M'=m, Y'=y}}{\prob{N=m,Y'=y}} \nonumber \\
    &\qquad {} + \sum_{(m,x,y)\not\in G_3} \prob{M'=m, X'=x, Y'=y} \cdot
    \log \frac{\prob{M'=m, Y'=y}}{\prob{N=m, Y'=y}} \nonumber \\
    &\geq -1 + \prob{\br{M',X',Y'} \notin G_3} \cdot \frac{\ve+1}{\ve'} \label{eqn:myc2} .
\end{align}
Above Eq.~\eqref{eqn:myc1} follows from Fact~\ref{fact:subsystem
monotone} and Eq.~\eqref{eqn:myc2} follows from definition of $G_3$.
This implies  $\prob{\br{M',X',Y'} \notin G_3} \leq \varepsilon'$.

On combining the bounds for the three terms, using
Eq.~\eqref{eqn:eq6} and using the union bound we get (recall $1 > \ve >0$)
\[  \pr{(m,x,y)\leftarrow M',X',Y'}
    {\log \frac{\prob{N=m|X'=x}}{\prob{N=m|Y'=y}}
    \geq \frac{c+5}{\ve'}}
    \leq 3 \ve'.\]
Now using $\onenorm{X'Y'M' - X'Y'N} \leq \sqrt{\ve} $ (as was shown previously), we finally have,
\[  \pr{(m,x,y)\leftarrow N,X',Y'}
    {\log \frac{\prob{N=m|X'=x}}{\prob{N=m|Y'=y}}
    \geq \frac{c+5}{\ve'}}
    \leq 3 \ve'+\sqrt{\ve}. \qedhere\]
\end{proof}

We will need the following further generalization of the previous
lemma.
\begin{lemma}\label{lem:simulationsimple}
Let $t \geq 1$ be an integer. Let $\ve'>0$, $c_s \geq 0, 1> \ve_s>0$
for each $1\leq s\leq t$. Let $R',X',Y',M_1',\ldots,M_t'$, be random
variables for which the following holds (below $M'_{<s} \defeq M'_1
\cdots M'_{s-1}$),
\[\condmutinf{X'}{M'_s}{Y'R'M'_{<s}}\leq c_s,  \quad \condmutinf{Y'}{M'_s}{X'R'M'_{<s}}\leq\ve_s, \quad \mbox{for odd $s$ }\]
and
\[\condmutinf{Y'}{M'_s}{X'R'M'_{<s}}\leq c_s, \quad \condmutinf{X'}{M'_s}{Y'R'M'_{<s}}\leq\ve_s, \quad \mbox{for even  $s$.}\]
There exists a public-coin $t$-message protocol $\P_t$ between
Alice, with input $X'R'$, and Bob, with input $Y'R'$, with Alice
sending the first message. The total communication is
$$\frac{\sum_{s=1}^tc_s+5t}{\ve'}+\O\(t\log\frac{1}{\ve'}\),$$
and at end of the protocol, both Alice and Bob possess random
variables $M_1, \ldots, M_t$, satisfying
\[\|R'X'Y'M_1\cdots M_t-R'X'Y'M_1'\cdots M_t'\|_1\leq 3\sum_{s=1}^t\sqrt{\ve_s}+6\ve't.\]
\end{lemma}
\begin{proof}
We prove the lemma by induction on $t$.  For the base  case $t=1$,
note that
$$\condmutinf{X'R'}{M_1'}{Y'R'}=\condmutinf{X'}{M_1'}{Y'R'} \leq c_1$$
 and
$$ \condmutinf{Y'R'}{M_1'}{X'R'}=\condmutinf{Y'}{M_1'}{X'R'} \leq \ve_1. $$
Lemma \ref{lem:generalized braverman rao}  implies (by taking $X',
Y, 'M'$ in Lemma \ref{lem:generalized braverman rao} to be $X'R',
Y'R', M_1'$ respectively) that Alice, with input $X'R'$, and Bob,
with input  $Y'R'$,  can run a public-coin protocol with a single
message from Alice to Bob of
$$\frac{c_1+5}{\ve'}+\O(\log\frac{1}{\ve'})$$
bits and generate a new random variable $M_1$ satisfying
$$\|R'X'Y'M'_1-R'X'Y'M_1\|_1\leq3\sqrt{\ve_1}+6\ve'.$$
Now let  $t > 1$. Assume $t$ is odd, for even $t$  a similar argument
will follow. From the induction hypothesis there exists a public-coin
$t-1$ message protocol $\mathcal{P}_{t-1}$ between Alice, with input
$X'R'$, and Bob, with input $Y'R'$, with Alice sending the first
message, and total communication
\begin{equation} \label{eqn:commt-1}
\frac{\sum_{s=1}^{t-1}c_s+5(t-1)}{\ve'}+\O\((t-1)\log\frac{1}{\ve'}\),
\end{equation}
such that at the end Alice and Bob both possess random variables
$M_1, \ldots, M_{t-1}$ satisfying
\begin{equation} \label{eqn:errort-1}
\|R'X'Y'M_1\cdots M_{t-1}-R'X'Y'M_1'\cdots M_{t-1}'\|_1\leq
3\sum_{s=1}^{t-1}\sqrt{\ve_s}+6\ve'(t-1).
\end{equation}
Note that
$$\condmutinf{Y'R'M_{<t}'}{M_t'}{X'R'M_{<t}'}=\condmutinf{Y'}{M_t'}{X'R'M_{<t}'} \leq c_t$$
and
$$\condmutinf{X'R'M_{<t}'}{M_t'}{Y'R'M_{<t}'}=\condmutinf{X'}{M_t'}{Y'R'M_{<t}'} \leq \ve_t. $$
Therefore Lemma \ref{lem:generalized braverman rao}  implies (by
taking $X', Y, 'M'$ in Lemma \ref{lem:generalized braverman rao} to
be $X'R'M_{<t}', Y'R'M_{<t}', M_t'$ respectively) that Alice, with
input $X'R'M_{<t}'$, and Bob, with input $Y'R'M_{<t}'$,  can run a
public coin protocol $\mathcal{P}$ with a single message from Alice
to Bob of
\begin{equation}
\label{eqn:commt}
\frac{c_t+5}{\ve'}+\O\(\log\frac{1}{\ve'}\)
\end{equation}
bits and generate a new random variable $M''_t$ satisfying
\begin{equation}\label{eqn:eq9}
\onenorm{R'X'Y'M_1'\cdots M_{t-1}'M_t'-R'X'Y'M_1'\cdots M_{t-1}'M''_t}\leq
3\sqrt{\ve_t}+6\ve'.
\end{equation}
Fact~\ref{fact:l1 monotone} and Eq.~\eqref{eqn:errort-1} imply that
Alice, on input $X'R'M_{<t}$ and Bob on input $Y'R'M_{<t}$, on
running the same protocol $\mathcal{P}$ will generate a new random
variable $M_t$ satisfying
\begin{align}
& \|R'X'Y'M_1\cdots M_{t-1}M_t-R'X'Y'M_1'\cdots M_{t-1}'M_t''\|_1 \nonumber \\
& = \|R'X'Y'M_1\cdots M_{t-1}-R'X'Y'M_1'\cdots M_{t-1}'\|_1 \nonumber \\
& \leq  3\sum_{s=1}^{t-1}\sqrt{\ve_s}+6\ve'(t-1). \label{eqn:myeq9}
\end{align}
Therefore by composing protocol $\mathcal{P}_{t-1}$ and protocol
$\mathcal{P}$ and using Equations~\eqref{eqn:commt-1},~\eqref{eqn:commt}, \eqref{eqn:eq9}, \eqref{eqn:myeq9} we get a
public-coin $t$-message protocol $\mathcal{P}_t$  between Alice,
with input $X'R'$, and Bob, with input $Y'R'$, with Alice sending
the first message, and total communication
\begin{equation*}
\frac{\sum_{s=1}^{t}c_s+5t}{\ve'}+\O\(t \log\frac{1}{\ve'}\),
\end{equation*}
such that at the end Alice and Bob both possess random variables
$M_1, \ldots, M_{t}$ satisfying
\begin{equation*}
\|R'X'Y'M_1\cdots M_{t}-R'X'Y'M_1'\cdots M_{t}'\|_1\leq
3\sum_{s=1}^{t}\sqrt{\ve_s}+6\ve' t. \qedhere
\end{equation*} 
\end{proof}
Following lemma, obtained from the lemma above, is the one that we
will finally use in the proof of our main result.
\begin{lemma}\label{lem:simulation}
Let random variables $R',X',Y',M_1',\ldots,M_t'$  and  numbers
$\ve', c_s,\ve_s$ satisfy all the conditions in Lemma
\ref{lem:simulationsimple}.  Let $\tau >0$ and  let random variables
$(X,Y)$ be $(1-\tau)$-embeddable in $(X'R',Y'R')$. There exists a
public-coin $t$-message protocol $\Q_t$ between Alice, with input
$X$, and Bob, with input $Y$, with Alice sending the first message,
and total communication
$$\frac{\sum_{s=1}^tc_s+5t}{\ve'}+\O\(t\log\frac{1}{\ve'}\)$$
bits, such that at the end Alice possesses $R_A M_1\cdots M_t$ and
Bob possesses $R_BM_1\cdots M_t$, such that
\[\|XYR_AR_BM_1\cdots M_t-X'Y'R'R'M_1'\cdots M_t'\|_1\leq\tau+3\sum_{s=1}^t\sqrt{\ve_s}+6\ve't.\]
\end{lemma}
\begin{proof}
In $\Q_t$, Alice and Bob, using public coins and no communication
first generate $R_A,R_B$ such that $\onenorm{XYR_AR_B-X'Y'R'R'}\leq\tau$.
They can do this from the Definition \ref{def:embeddable} of
embedding.  Now they will run protocol $\P_t$ (as in  Lemma
\ref{lem:simulationsimple}) with Alice's input being $XR_A$ and
Bob's input being $YR_B$ and at the end both possess  $M_1,\ldots,
M_t$. From  Lemma~\ref{lem:simulationsimple}, the communication of
$\Q_t$ is as desired. Now from Fact~\ref{fact:l1 monotone} and
Lemma~\ref{lem:simulationsimple}
\[\|XYR_AR_BM_1\cdots M_t-X'Y'R'R'M_1'\cdots M_t'\|_1\leq\tau+3\sum_{s=1}^t\sqrt{\ve_s}+6\ve't. \qedhere\] 
\end{proof}
We are now ready to prove our main result, Theorem~\ref{thm:direct
product}. We restate it here for convenience.
\begin{maintheorem}
    Let  $\X$, $\Y$, $\Z$ be finite sets,
    $f\subseteq\X\times\Y\times\Z$ a relation, $\ve > 0$ and $k, t \geq1$ be integers. There exists a constant $\kappa$ such that,
    \[ \mathrm{R}^{(t),\mathrm{pub}}_{1-(1-\ve/2)^{\Omega(k\ve^2/t^2)}}(f^k) = \Omega \( \frac{\ve \cdot k}{t} \cdot \(\mathrm{R}^{(t),\mathrm{pub}}_{\ve}(f) - \frac{\kappa t^2}{\ve^2} \) \).\]
\end{maintheorem}
\begin{proofof}{Theorem~\ref{thm:direct product}}
 Let $c \defeq \mathrm{D}^{(t), \mu}_{\ve}(f) - \frac{\kappa t^2}{\ve^2} $ for $\kappa$ to be chosen later. Let $\delta \defeq \frac{\ve^2}{7500t^2}$ and $\delta_1 = \frac{\ve}{3000 t}$. From Yao's principle,
Lemma~\ref{lem:yaos principle}, it suffices to prove that for any
distribution $\mu$ on $\X\times \Y$,
$$\mathrm{D}^{(t),\mu^k}_{1-(1-\ve/2)^{\lfloor \delta k \rfloor}}(f^k) \geq \delta_1 k  c \enspace .$$
Let $XY\sim\mu^k$. Let $\Q$ be a $t$-message deterministic protocol
between Alice, with input $X$, and Bob, with input $Y$, that
computes $f^k$, with Alice sending the first message and total
communication $\delta_1 k c$ bits. We
assume $t$ is odd for the rest of the argument and Bob makes the
final output (the case when $t$ is even follows similarly). The
following Claim \ref{claim:next coordinate identification} implies
that the success of $\Q$ is at most $(1-\ve/2)^{\lfloor \delta k
\rfloor}$ and this shows the desired.
\end{proofof}
\newcommand {\Tsr} {\ensuremath{ T^{\br{r}} }}
\begin{claim}\label{claim:next coordinate identification}
    For each $i\in[k]$, define a binary random variable $T_i\in\{0,1\}$, which
represents the success of $\Q$ (that is Bob's output being correct)
on the $i$-th instance. That is, $T_i=1$ if the $\Q$ computes the
$i$-th instance of $f$ correctly, and $T_i=0$ otherwise. Let $k'
\defeq \lfloor\delta k\rfloor $.
    There exist $k'$ coordinates \set{i_1, \ldots, i_{k'}} such that
    for each $1 \leq r \leq k'-1$, either
    \[ \prob{\Tsr = 1} \leq(1-\ve/2)^{k'} \]
    or
    \[ \prob{ T_{i_{r+1}}=1 \middle\vert \Tsr = 1 }
    \leq1-\ve/2 ,\]
    where $ \displaystyle \Tsr \defeq \prod_{j=1}^{r} T_{i_j} $.
\end{claim}
\begin{proofof}{Claim~\ref{claim:next coordinate identification}}
For $s \in [t]$, denote the $s$-th message of $\Q$ by $M_s$. Define $M\defeq M_1\cdots M_t$.
% and $c_s\defeq|M_s|/(\delta_1 k)$ (where $|M_s|$ is the number of bits in $M_s$). Note that $\sum_{s=1}^tc_s=c$. 
In the following we assume $1
\leq r < k'$, however same arguments also work when $r=0$, that is
for identifying the first coordinate, which we skip for the sake of
avoiding repetition. Suppose we have already identified $r$
coordinates $i_1,\ldots,i_r$ satisfying that
$\Pr[T_{i_1}=1]\leq1-\ve/2$ and
$\Pr[T_{i_{j+1}}=1|T^{(j)}=1]\leq1-\ve/2$ for $1\leq j\leq r-1$. If
$\prob{\Tsr = 1} \leq (1-\ve/2)^{k'}$, we are done. So from now on,
assume $\prob{\Tsr=1} > (1-\ve/2)^{k'} \geq 2^{-\delta k}$.

Let $D$ be a random variable uniformly distributed in $\{0,1\}^k$
and independent of $XY$. Let $U_i=X_i$ if $D_i=0$, and $U_i=Y_i$ if
$D_i=1$. For any random variable $L$, let us introduce the notation:
$L^1 \defeq (L|\Tsr=1)$. For example, $X^1Y^1=(XY|\Tsr=1)$. If
$L=L_1\cdots L_k$, define $L_{-i}\defeq L_1\cdots
L_{i-1}L_{i+1}\cdots L_k$, and $L_{<i}\defeq L_1\cdots L_{i-1}$.
Random variable $L_{\leq i}$ is defined
analogously.  Let $C
\defeq \set{i_1, \ldots, i_r}$. Define $R_i\defeq
D_{-i}U_{-i}X_{C\cup[i-1]}Y_{C\cup [i-1]}$ for $i\in[k]$. We denote
an element from the range of $R_i$ by $r_i$.

To prove the claim, we will show that there exists a coordinate
$j\not\in C$ such that,
\begin{enumerate}

  \item $(X_jY_j)$ can be embedded well in $(X^1_jR^1_j, Y^1_jR^1_j)$.

  \item Random variables $X^1_j, Y^1_j, M^1_1, \ldots, M^1_t$ satisfy the conditions of Lemma~\ref{lem:simulationsimple} with appropriate parameters.

\end{enumerate}
Following is helpful in meeting the first condition.
\begin{align}
    \delta k &> \rminent{X^1Y^1}{XY}
        \label{eqn:eq1} \\
    &\geq \relent{X^1Y^1}{XY} \quad \geq \quad \sum_{i\notin C} \relent{X^1_iY^1_i}{X_i Y_i} ,
    \label{eqn:one coordinate close}
\end{align}
where Eq.~\eqref{eqn:eq1} follows from the assumption that
$\prob{\Tsr=1} > 2^{-\delta k}$, and Eq.~\eqref{eqn:one coordinate
close} is from Fact~\ref{fact:relative entropy splitting}. Also
consider,
\begin{align}
    \delta k &> \rminent{X^1Y^1D^1U^1}{XYDU} \nonumber \\
    &\geq \relent{X^1Y^1D^1U^1}{XYDU} \nonumber \\
    &\geq \expec{(d,u,x_C,y_C)\leftarrow D^1,U^1,X^1_C,Y^1_C}
        {\relent{\br{X^1Y^1}_{d,u,x_C,y_C}}{\br{XY}_{d,u,x_C,y_C}}}
        \label{eqn:eq2} \\
    &= \sum_{i \notin C} \; \expec{(d,u,x_{C\cup[i-1]},y_{C\cup[i-1]})\\
        \leftarrow D^1,U^1,X_{C\cup[i-1]}^1,Y_{C\cup[i-1]}^1}
        {\relent{\br{X_i^1Y_i^1}_{d,u,x_{C\cup[i-1]},y_{C\cup[i-1]}}}
        {\br{X_iY_i}_{d,u,x_{C\cup[i-1]},y_{C\cup[i-1]}}}}
        \label{eqn:eq3} \\
    &=\sum_{i\notin C}
        \expec{(d_i,u_i,r_i) \leftarrow D^1_i,U^1_i,R^1_i}
        { \relent{(X^1_iY^1_i)_{d_i,u_i,r_i}}{(X_iY_i)_{d_i,u_i,r_i}}}
        \label{eqn:eq4}\\
    &= \frac{1}{2} \sum_{i \notin C} \;
        \expec{(r_i,x_i)\leftarrow R^1_i,X^1_i}
        {\relent{\br{Y_i^1}_{r_i, x_i}}
        {\br{Y_i}_{x_i}}} +  \frac{1}{2} \sum_{i \notin C} \;
    \expec{(r_i,y_i)\leftarrow R^1_i,Y^1_i}
    {\relent{\br{X_i^1}_{r_i, y_i}}
    {\br{X_i}_{y_i}}} .
        \label{eqn:junk is independent for Y}
\end{align}
Above, Eq.~\eqref{eqn:eq2} and Eq.~\eqref{eqn:eq3} follow from
Fact~\ref{fact:relative entropy splitting}; Eq.~\eqref{eqn:eq4} is
from the definition of $R_i$. Eq.~\eqref{eqn:junk is independent for
Y} follows since $D^1_i$ is independent of $R^1_i$ and with
probability half $ D^1_i$ is $0$, in which case $U^1_i = X^1_i$  and
with probability half $ D^1_i$ is $1$ in which case $U_i^1 = Y_i^1$.

Following calculations are helpful in meeting the second condition.
\begin{align}
    \delta_1 c k &\geq \abs{M^1} \nonumber \\
    &\geq \condmutinf{X^1Y^1}{M^1}{D^1U^1X^1_CY^1_C} \nonumber \\
    &=\sum_{i\notin C} \condmutinf{X^1_iY^1_i}{M^1}
    {D^1U^1X^1_{C\cup[i-1]}Y^1_{C\cup[i-1]}}
    \nonumber \\
    &=\sum_{i\notin C}\sum_{s=1}^t\condmutinf{X^1_iY^1_i}{M^1_s}
    {D^1U^1X^1_{C\cup[i-1]}Y^1_{C\cup[i-1]}M^1_{<s}}
    \nonumber \\
    &= \sum_{i\notin
    C}\sum_{s=1}^t\condmutinf{X^1_iY^1_i}{M^1_s}{D^1_iU^1_iR^1_iM^1_{<s}}\nonumber\\
    & = \sum_{i\notin
    C} \br{\sum_{s \text{ odd}}
    \condmutinf{X^1_iY^1_i}{M^1_s}{D^1_iU^1_iR^1_iM^1_{<s}}+\sum_{s\text{ even}}
    \condmutinf{X^1_iY^1_i}{M^1_s}{D^1_iU^1_iR^1_iM^1_{<s}}}\nonumber\\
    & \geq \frac{1}{2}\sum_{i\notin C}\br{\sum_{s\text{ odd}}\condmutinf{X^1_i}{M^1_s}{R^1_iY^1_iM^1_{<s}}+\sum_{s\text{
even}}\condmutinf{Y^1_i}{M^1_s}{R^1_iX^1_iM^1_{<s}}}.
    \label{eqn:mutual inf XY and M}
\end{align}
Above we have used the chain rule for mutual information several times.
Last inequality follows since $D^1_i$ is independent of
$(X^1_iY_i^1R^1_iM^1)$ and with probability half $ D^1_i$ is $0$, in
which case $U^1_i = X^1_i$  and with probability half $ D^1_i$ is
$1$ in which case $U_i^1 = Y_i^1$.

For the following, let $s \in [t]$ be odd.
\begin{align}
    \delta k &\geq \rminent{D^1 U^1 X^1 Y^1 M^1_{\leq s}}
    {DUXYM_{\leq s}} \nonumber \\
    &\geq \relent{D^1 U^1 X^1 Y^1 M^1_{\leq s}}
    {DUXYM_{\leq s}} \nonumber \\
    &\geq \expec{(d,u,x_C,y_C,m_{\leq s})\leftarrow
    D^1,U^1,X^1_C,Y^1_C,M^1_{\leq s}}{\relent{(X^1Y^1)_{d,u,x_C,y_C,m_{\leq s}}}{(XY)_{d,u,x_C,y_C,m_{\leq s}}}} \nonumber\\
    &=\sum_{i\notin C} \expec{(d,u,x_{C\cup[i-1]},y_{C\cup [i-1]},m_{\leq
    s})\\ \leftarrow D^1,U^1,X^1_{C\cup [i-1]},Y^1_{C\cup [i-1]},M^1_{\leq s}}{\relent{(X^1_iY^1_i)_{d,u,x_{C\cup [i-1]},y_{C\cup [i-1]},m_{\leq s}}}{(X_iY_i)_{d,u,x_{C\cup [i-1]},y_{C\cup [i-1]},m_{\leq
    s}}}} \nonumber\\
    &= \sum_{i\notin C}  \expec{(d_i,u_i,r_i,m_{\leq s})\leftarrow D^1_i,U^1_i,R^1_i,M^1_{\leq s}} {\relent{(X^1_iY^1_i)_{d_i,u_i,r_i,m_{\leq s}}}{(X_iY_i)_{d_i,u_i,r_i,m_{\leq
    s}}}}\label{eqn:c1}\\
    & \geq \frac{1}{2}\sum_{i\notin C}  \expec{(x_i,r_i,m_{\leq s})\leftarrow X^1_i,R^1_i,M^1_{\leq s}}
    {\relent{(Y^1_i)_{x_i,r_i,m_{\leq s}}}{(Y_i)_{x_i,r_i,m_{\leq
    s}}}}\nonumber\\
    &= \frac{1}{2}\sum_{i\notin C}  \expec{(x_i,r_i,m_{\leq s})\leftarrow X^1_i,R^1_i,M^1_{\leq s}}
    {\relent{(Y^1_i)_{x_i,r_i,m_{\leq s}}}{(Y_i)_{x_i,r_i,m_{<
    s}}}} \label{eqn:c2}\\
    & = \frac{1}{2}\sum_{i\notin C}  \expec{(x_i,r_i,m_{< s})\leftarrow X^1_i,R^1_i,M^1_{<
    s}}{\relent{(Y^1_iM^1_s)_{x_i,r_i,m_{<s}}}{(Y_i)_{x_i,r_i,m_{<s}}\otimes
    (M^1_s)_{x_i,r_i,m_{<s}}}}\nonumber\\
    &\geq \frac{1}{2}\sum_{i\notin C}  \expec{(x_i,r_i,m_{< s})\leftarrow X^1_i,R^1_i,M^1_{<
    s}}{\mutinf{(Y^1_i)_{x_i,r_i,m_{<s}}}{(M^1_s)_{x_i,r_i,m_{<s}}}}\label{eqn:c3}\\
    & = \frac{1}{2}\sum_{i\notin C}\condmutinf{Y^1_i}{M^1_s}{X^1_iR^1_iM^1_{<s}}.
    \label{eqn:Markovmutual1}
\end{align}

Above we have used Fact~\ref{fact:relative entropy splitting} several times.
Eq.~\eqref{eqn:c1} follows from the definition of $R_i$;
Eq.~\eqref{eqn:c2} follows from the fact that $Y\leftrightarrow
X_iR_iM_{<s}\leftrightarrow M_s$ for any $i$, whenever $s$ is odd;
Eq.~\eqref{eqn:c3} follows from Fact \ref{fact:mutinf is min}.

From a symmetric argument, we can show that when $s\in [t]$ is even,

\begin{equation}\label{eqn: Markovmutual2}
\frac{1}{2}\sum_{i\notin
C}\condmutinf{X^1_i}{M^1_s}{Y^1_iR^1_iM^1_{<s}}\leq \delta k.
\end{equation}
Eq.~\eqref{eqn:Markovmutual1} and Eq.~\eqref{eqn: Markovmutual2}
together imply

\begin{equation}\label{eqn:Markov2}
\sum_{i\notin C}\br{\sum_{s\text{
odd}} \condmutinf{Y^1_i}{M^1_s}{R^1_iX^1_iM^1_{<s}}+\sum_{s\text{
even}}\condmutinf{X^1_i}{M^1_s}{R^1_iY^1_iM^1_{<s}}}\leq 2\delta kt.
\end{equation}

Combining Equations (\ref{eqn:one coordinate close})(\ref{eqn:junk
is independent for Y})(\ref{eqn:mutual inf XY and
M})(\ref{eqn:Markov2}), and making standard use of Markov's
inequality, we can get a coordinate $j\notin C$ such that
\begin{eqnarray}
&&\relent{X^1_jY^1_j}{X_jY_j}\leq12\delta, \label{eqn:one coordinate closeget}\\
&&\expec{(r_j,x_j)\leftarrow R^1_j, X^1_j}
        {\relent{\br{Y_j^1}_{ r_j,x_j}}
        {\br{Y_j}_{x_j}}}\leq12\delta, \label{eqn:junk is independent for Xget}\\
&&\expec{(r_j,y_j)\leftarrow R^1_j, Y^1_j}
        {\relent{\br{X_j^1}_{r_j,y_j}}
        {\br{X_j}_{y_j}}}\leq12\delta\label{eqn:junk is independent for
        Yget},\\
&&\sum_{s \text{
odd}}\condmutinf{X^1_j}{M^1_s}{R^1_jY^1_jM^1_{<s}}+\sum_{s\text{
even}}\condmutinf{Y^1_j}{M^1_s}{R^1_jX^1_jM^1_{<s}}\leq 12\delta_1 c , \label{eqn:mutual inf XYGET and M}\\
&&\sum_{s \text{
odd}}\condmutinf{Y^1_j}{M^1_s}{R^1_jX^1_jM^1_{<s}}+\sum_{s \text{
even}}\condmutinf{X^1_j}{M^1_s}{R^1_jY^1_jM^1_{<s}}\leq12\delta t.\label{eqn:MarkovGet}
\end{eqnarray}
Set $\ve'\defeq\frac{\ve}{125t}$, and
\[\ve_s\defeq\begin{cases}\condmutinf{Y^1_j}{M^1_s}{R^1_jX^1_jM^1_{<s}} &\text{$s\in[t] $ odd,}\\
\condmutinf{X^1_j}{M^1_s}{R^1_jY^1_jM^1_{<s}} &\text{$s\in[t]$ even.}
\end{cases}\] 
\[c_s\defeq\begin{cases}\condmutinf{Y^1_j}{M^1_s}{R^1_jX^1_jM^1_{<s}} &\text{$s\in[t] $ even,}\\
\condmutinf{X^1_j}{M^1_s}{R^1_jY^1_jM^1_{<s}} &\text{$s\in[t]$ odd.}
\end{cases}\] 
By \eqref{eqn:MarkovGet},
$\sum_{s=1}^t\sqrt{\ve_s}\leq\sqrt{12\delta} t $. From Equations
(\ref{eqn:one coordinate closeget})(\ref{eqn:junk is independent for
Xget})(\ref{eqn:junk is independent for Yget}) and Lemma
\ref{lem:relent embeddable} we can infer that  $(X_jY_j)$ is
$(1-10\sqrt{3\delta })$-embeddable in $(X^1_jR^1_j;Y^1_jR^1_j)$.
This, combined with Equations \eqref{eqn:mutual inf XYGET and
M}\eqref{eqn:MarkovGet} and Lemma \ref{lem:simulation} (take $\ve', \ve_s, c_s$ in the lemma to be as defined above and take $XYX'Y'R'M_1' \cdots M_t'$ in the lemma to be $X_jY_jX_j^1Y^1_jR^1_jM_1^1 \cdots M_t^1$) imply the
following (for appropriate constant $\kappa$). There exists a public-coin $t$-message protocol $\Q^1$
between Alice, with input $X_j$, and Bob, with input $Y_j$, with
Alice sending the first message and total communication,
 $$\frac{12\delta_1 c+5t}{\ve'}+\O(t\log\frac{1}{\ve'})<  \mathrm{D}^{(t), \mu}_{\ve}(f)  ,$$
such that at the end Alice possesses $R_AM_1 \cdots M_t$  and Bob
possesses $R_BM_1 \cdots M_t$, satisfying
$$ \onenorm{X_j Y_j R_A R_B M_1 \cdots M_t - X^1_j Y^1_j R^1_j R^1_j M^1_1 \cdots M^1_t} \leq 10\sqrt{3\delta } +3 \sqrt{12\delta} t + 6 \ve' t < \ve/2. $$
 Assume for contradiction that $\prob{ T_{j}=1 \middle\vert \Tsr = 1 }    > 1-\ve/2$.  Consider a protocol $\Q^2$ (with no communication) for $f$ between Alice, with input $X^1_j  R^1_j M^1_1 \cdots M^1_t$, and Bob, with input $Y^1_jR^1_j M^1_1 \cdots M^1_t$, as follows. Bob generates the rest of the random variables present in $Y^1$ (not present in his input) himself since,  conditioned on his input,  those other random variables are independent of Alice's input (here we use Fact~\ref{fact:commindep}). Bob then generates the output for the $j$-th coordinate in $\Q$, and makes it the output of $\Q^2$. This ensures that the success probability of Bob in $\Q^2$ is  $\prob{ T_{j}=1 \middle\vert \Tsr = 1 }    >1 - \ve/2$.  Now consider protocol $\Q^3$ for $f$, with Alice's input $X_j$ and Bob's input $Y_j$, which is a composition of $\Q^1$ followed by $\Q^2$. This ensures, using Fact~\ref{fact:l1 monotone}, that success probability of Bob (averaged over public coins and the inputs $X_jY_j$) in $\Q^3$ is larger than $1 - \ve$. Finally by fixing the public coins of $\Q^3$,  we get a deterministic protocol $\Q^4$ for $f$ with Alice's input $X_j$ and Bob's input $Y_j$ such that the communication of $\Q^4$ is less than  $\mathrm{D}^{(t), \mu}_{\ve}(f)$ and Bob's success probability (averaged over the inputs $X_jY_j$) in $\Q^4$ is larger than $1 - \ve$. This is a contradiction to the definition of $\mathrm{D}^{(t), \mu}_{\ve}(f)$ (recall that $X_jY_j$ are
distributed according to $\mu$). Hence it must be that $\prob{
T_{j}=1 \middle\vert \Tsr = 1 }    \leq 1-\ve/2$. The claim now
follows by setting $i_{r+1} = j$.
\end{proofof}

\subsection*{Open problems}
\label{sec:conclusion}

Some natural questions that arise from this work are:

\begin{enumerate}

  \item Can the dependence on $t$ in our direct product theorem be improved?

  \item Can these techniques be extended to show direct product theorems for bounded-round
  quantum communication complexity?

\end{enumerate}

\bibliographystyle{alpha}
\bibliography{ref}
\end{document}